\documentclass[fleqn]{llncs}
\usepackage{etex}
\usepackage{graphicx}
\usepackage{pgfplots}
\pgfplotsset{width=12.5cm, height=8cm}

\usepackage{fancyhdr}
\usepackage{array}
\usepackage{datetime}
\usepackage{paralist}
\usepackage{caption}
\usepackage{subcaption}
\usepackage{bbm}
\usepackage{pstricks}
\usepackage{colortbl}
\usepackage{enumitem}
\usepackage{color,soul}
\usetikzlibrary{patterns}
\usepgfplotslibrary{groupplots}
\usetikzlibrary{patterns}
\usepackage{color}
\usepackage{pst-node}
\usepackage{pst-grad}
\usepackage{pst-coil}
\usepackage{pst-text}
\usepackage{pst-rel-points}
\usepackage{pstricks-add}
\usepackage{amssymb}
\usepackage{amsmath}
\usepackage{mathptm}
\usepackage{fancyvrb}
\usepackage{calc}
\usepackage{etex}
\usepackage{ifthen}
\usepackage{multirow}
\usepackage{rotating}
\usepackage{longtable}
\usepackage[normalem]{ulem}
\usepackage{stmaryrd}
\usepackage{hyperref}
\usepackage{xspace}
\usepackage{mathtools}
\usepackage{datetime}
\usepackage{algorithm}
\usepackage[noend]{algpseudocode}

\usepackage{xr}
\usepackage{bm}

\usepackage{tikz}
\usepackage{listings}
\usepackage{color}
\newdateformat{titledate}{\THEDAY\ \monthname[\THEMONTH], \THEYEAR}
\newdateformat{mydate}{\monthname[\THEMONTH] \THEYEAR}

\protect{\newlength{\BRACEL}}
\protect{\newlength{\TAL}} 
\newlength{\codeLineLength}

\newcommand{\transform}{\text{\em transform\/}\xspace}

\newcommand{\cskw}{\text{\bf \{\/}\xspace}
\newcommand{\cekw}{\text{\bf \}\/}\xspace}
\newcommand{\qkw}{\text{\bf ?\/}\xspace}
\newcommand{\ckw}{\text{\bf :\/}\xspace}
\newcommand{\ifkw}{\text{\bf if\/}\xspace}
\newcommand{\forkw}{\text{\bf for\/}\xspace}

\newcommand{\assertkw}{\text{\bf assert\/}\xspace}

\newcommand{\elsekw}{\text{\bf else\/}\xspace}

\newcommand{\arrvar}{\text{$a$}\xspace}
\newcommand{\absind}{\text{$i_a$}\xspace}
\newcommand{\absarr}{\text{$x_a$}\xspace}

\newcommand{\ms}{\text{$\sigma$}\xspace}

\newcommand{\msl}{\text{$\sigma_l$}\xspace}
\newcommand{\msm}{\text{$\sigma_m$}\xspace}

\newcommand{\expr}{\text{$e$}\xspace}
\newcommand{\sexpr}{\text{$\mathbbm{E}$}\xspace}
\newcommand{\var}{\text{$x$}\xspace}

\newcommand{\size}{\text{\sf\em lastof}\xspace}
\newcommand{\loc}{\text{$\ell$}\xspace}
\newcommand{\svar}{\text{$\mathbb{V}$}\xspace}

\newcommand{\earr}{\text{$\mathbb{E}_A$}\xspace}

\newcommand{\eval}[2]{\text{$\left\llbracket #1 \right\rrbracket_{#2}$}\xspace}
\newcommand{\sval}{\text{$\mathbb{C}$}\xspace}
\newcommand{\vexpr}{\text{\eval{e}{\ms}}\xspace}

\newcommand{\pc}{\text{$P$}\xspace}
\newcommand{\pa}{\text{$P'$}\xspace}
\newcommand{\msa}{\text{$\sigma'$}\xspace}

\newcommand{\msalp}{\text{$\sigma'\!\!_{l'}$}\xspace}
\newcommand{\msamp}{\text{$\sigma'\!\!_{m'}$}\xspace}

\renewcommand{\int}{\text{\sf I}\xspace}

\newcommand{\smem}{\text{$\mathbb{M}$}\xspace}
\newcommand{\smema}{\text{$\mathbb{M'}$}\xspace}

\newcommand{\assert}{\text{$A_n$}\xspace}
\newcommand{\assertinv}{\text{$A\!_n^{inv}$}\xspace}

\newcommand{\emit}{\text{\sf\em emit}\xspace}
\newcommand{\fullarrayaccess}{\text{\sf\em fullarrayaccess}\xspace}
\newcommand{\loopbound}{\text{\sf\em loopbound}\xspace}
\newcommand{\loopdefs}{\text{\sf\em loopdefs}\xspace}
\newcommand{\T}{\text{\sf\em transform}\xspace}

\newcommand{\Prog}{\text{\sf\em P}\xspace}
\newcommand{\Stmt}{\text{\sf\em S}\xspace}
\newcommand{\Expr}{\text{\sf\em E}\xspace}
\newcommand{\Lvalue}{\text{\sf\em L}\xspace}
\newcommand{\Init}{\text{\sf\em I}\xspace}

\begin{document}

\pagestyle{plain}
\setcounter{page}{1}
\pagenumbering{arabic}


\title{Scaling Bounded Model Checking By Transforming Programs With Arrays}

\author{Anushri Jana \inst{1} \and Uday P. Khedker\inst{2} \and Advaita Datar \inst{1} \and R. Venkatesh\inst{1} \and Niyas C \inst{1}
}
\institute{ Tata Research Development and Design Centre, Pune, India\\
\email{\{anushri.jana,advaita.datar,r.venky,niyas.c\}@tcs.com}
\and Indian Institute of Technology Bombay, India. \\
\email{uday@cse.iitb.ac.in}
}

\maketitle
\begin{abstract}
Bounded Model Checking is one the most successful techniques for finding bugs in program. 
However, for programs with loops iterating over large-sized
arrays, bounded model checkers often exceed the limit of resources available to
them.  We present a transformation that enables bounded model checkers to
verify a certain class of array properties. Our technique transforms an
array-manipulating program in \textsc{Ansi}-C to an array-free and loop-free
program. The transformed program can efficiently be verified by an
off-the-shelf bounded model checker. Though the transformed program is, in
general, an abstraction of the original program, we formally characterize the
properties for which the transformation is precise. We demonstrate the
applicability and usefulness of our technique on both industry code as well as
academic benchmarks.
\end{abstract}

\begin{keywords} Program Transformation, Bounded Model Checking, Array, Verification. 
\end{keywords}

\section{Introduction} 
\label{introduction} 

Bounded Model Checking is one of the most successful techniques for
finding bugs~\cite{copty2001benefits} as evidenced by success 
achieved by tools implementing this technique in verification
competitions~\cite{SVCOMP16,SVCOMP15}. Given a program P
and a property $\varphi$, Bounded Model Checkers (BMCs) unroll the loops in P a fixed
number of times and search for violations to $\varphi$ in the unrolled program.
However, for programs with loops of large or unknown bounds, bounded model
checking instances often exceed the limits of resources available. In our
experience, programs manipulating large-sized arrays invariably have such loops
iterating over indices of the array. Consequently, BMCs
routinely face the issue of scalability in proving properties on arrays. The
situation is not different even when the property is an {\it array invariant}
i.e., it holds for every element of the array, a characteristic which can
potentially be exploited for efficient bounded model checking.

Consider the example in Figure~\ref{Mex1} manipulating an array of structures
$a$. The structure has two fields, $p$ and $q$, whose values are assigned in
the first {\it for} loop (lines 8--13) such that $a[i].q$ is the square of
$a[i].p$ for every index $i$. The second {\it for} loop (lines 14--17) asserts
that this property indeed holds for each element in $a$. This is a {\it
safe} program i.e., none of the assertions admit a counterexample.
CBMC~\cite{cbmc}, a bounded model checker for C, in an attempt to unwind first
loop 100000 times, runs out of memory before it even reaches
the loop with assertion. In fact, we tried this example with several other
model checkers\footnote{Result for \emph{motivatingExample.c} at \htmladdnormallink{https://sites.google.com/site/datastructureabstraction/}{https://sites.google.com/site/datastructureabstraction/}} 
	and none of them
were able to prove this property because of large loop bounds.

\begin{figure*}[t]
\scriptsize
\begin{minipage}{0.45\textwidth}
\begin{lstlisting}[frame=single,language=c,escapeinside={(*@}{@*)},basicstyle=\ttfamily]
1. struct S {
2.    unsigned int p;		  
3.    unsigned int q;		 		  
4. } a[100000];
5. int i,k;

06. main()
07. {
08.  for(i=0; i<100000; i++)
09.  {
10.   k = i; 
11.   a[i].p = k;
12.   a[i].q = k * k ;
13.  }

14.  for (i=0; i<100000; i++)
15.  {
16.   (*@ \color[rgb]{0.75,0.164,0.164}{assert}@*)(a[i].q == 
          a[i].p * a[i].p);
17.  }
18. }
\end{lstlisting}
\caption{Motivating Example}
\label{Mex1}
\end{minipage}
\quad
\quad
\begin{minipage}{0.50\textwidth}
\begin{lstlisting}[frame=single,language=c,escapeinside={(*@}{@*)},basicstyle=\ttfamily]
1. struct S{
2.    unsigned int p;		  
3.    unsigned int q;		 		  
4. }x_a;
5. int i_a;
6. int i,k;
7. main()
8. {
9.   i_a = nd(0,99999);

//first loop body
10.  k = nd(0,100000);
11.  i = i_a;
12.  k = i;
13.  (i == i_a)? x_a.p = k : k;
14.  (i == i_a)? x_a.q = k * k : k*k ;
15.  k = nd(0,100000);
  
//second loop body
16.  i = i_a;
17.  (*@ \color[rgb]{0.75,0.164,0.164}{assert}@*)(((i==i_a)?x_a.q:nd())
        ==((i==i_a)?x_a.p:nd())
        *((i==i_a)?x_a.p:nd()));
18. }

\end{lstlisting} 
\caption{Transformed Code}
\label{TEx1}
\end{minipage}
\end{figure*}

One of the ways of proving this example safe is to show that the property holds
for any arbitrary element of the array, say at index $i_c$. This allows us to
get rid of those parts of the program that do not update $a[i_c]$ which, in
turn, eliminates the loop iterating over all the array indices. This enables
CBMC to verify the assertion without getting stuck in the loop unrollings.
Moreover, since $i_c$ is chosen nondeterministically from the indices of $a$,
the property holds for every array element without loss of generality.

This paper presents the transformation sketched above with the aim that the
transformed program is easier for a BMC to verify as compared
to the original program. The transformation
is over-approximative i.e.,  it give more values than that by the original program.
This ensures that if the original program is safe with respect to the chosen
property, so is the transformed program. However, the over-approximation
raises two important questions spanning practical and intellectual considerations:
\begin{enumerate}
\item {\it Is the proposed approach practically useful? Does the
transformation enable a BMC to verify real-world programs,
and even academic benchmarks, fairly often?}

We provide an answer to this through an extensive experimental evaluation over
industry code as well as 
examples in the array category of SV-COMP
2016 benchmarks. In all the cases, we show that our approach helps CBMC to scale.
We further demonstrate the applicability of
our technique to successfully identify a large number of false warnings (on an average 73\%)
reported by a static analyzer on arrays in large programs.

\item {\it Is it possible to characterize a class of properties for which it is
precise?}

In order to address this we provide a formal characterization of properties for
which the transformation is precise i.e., we state criteria under which
the transformed program is unsafe only when the original program is unsafe
(Section~\ref{sec::prec}).
\end{enumerate}

To summarize, this paper makes the following contributions:
\begin{itemize}
\item A new technique using the concept of \emph{witness index} that enables BMCs to verify array invariant
properties in programs with loops iterating over large-sized arrays. 
\item A formal characterization of properties for which the technique is
precise.
\item A transformation engine implementing the technique.
\item An extensive experimental evaluation showing the applicability of our
technique to real-world code as well as to academic benchmarks.
\end{itemize}
		
The rest of the paper starts with an informal description of the transformation
(Section~\ref{desc}) before we define the semantics
(Section~\ref{sec:dsa}) and formally state the transformations rules
(Section~\ref{transformation}).  Section~\ref{sec:proof} and ~\ref{sec::prec},
resp., describe the soundness and precision of our approach.  Section~\ref{exp}
presents the experimental setup and results. We discuss the related work in
Section~\ref{relwork} before concluding in Section~\ref{conclusion}.

\section{Informal Description}
\label{desc}

Given a program \pc containing loops iterating over an array \arrvar, we transform it 
to a program \pa that has a pair \text{$\langle \absarr, \absind \rangle$} of a
\emph{witness variable} and a \emph{witness index} for the array and the index
such that
$\absarr$ represents the element $a[i_a]$ of the original program. Further, loops are replaced by 
their customized bodies that operate only on $\absarr$ instead of all elements of $a$.

To understand the intuition behind our transformation, consider a trace $t$ of \pc ending on the assertion
\assert. Consider the last occurrence of a statement \text{$s: a[e_1]=e_2$} in $t$. 
We wish to transform \pc such that
there exists a trace $t'$ of \pa ending on \assert with value of \absind equal to that of $e_1$
and value of \absarr equal to that of $e_2$.  
%
%
We achieve this by transforming the program such that: 
\begin{itemize}
	\item \absind gets a non-deterministic value at the start of the program (this facilitates arbitrary choice of array element $a[\absind]$).
    	\item array writes and reads for $a[\absind]$ gets replaced with witness variable \absarr.   
	\item array writes other than $a[\absind]$ gets eliminated and reads gets replaced with non-deterministic value.
	\item loop body is executed only once either non-deterministically or unconditionally based on loop characteristics. During the
              execution of the loop body, 
		\begin{itemize}
			\item the loop iterator variable gets the value of \absind or a non-deterministic value (depending on loop characteristics), and
		\item all other scalar variables whose values may be different in different iterations 
              		gets non-deterministic values.
	\end{itemize}
\end{itemize}

Figure~\ref{TEx1} shows the transformed program \pa for the program \pc of Figure~\ref{Mex1}. 
Function \texttt{nd(l,u)} returns a non-deterministic value in the range $[l..u]$.
In \pa, the witness index $i\_a$ for array $a$ is globally assigned a 
non-deterministic value within the range of array size (at line 9).
In a run of BMC, the assertion is checked for this non-deterministically chosen element $a[\absind]$.
To ensure that values for the same index $a[i_a]$ are written and read, we replace 
array accesses by the witness variable \emph{\text{$x\_a$}} only when the value of index $i$ matches with $i\_a$ (lines 13, 14 and 17).
We remove loop header but retain loop body. 
To over-approximate the effect of removal of loop iterations we add non-deterministic assignments to all variables modified in the loop body, 
at the start of the transformed loop body and also
after the transformed loop body (lines 11 and 15). 
Note that we retain the original assignment statements too (line 12).
Since the loops at line 8 and line 14 in the original program iterate over the entire array, we equate loop iterator variable $i$ to
$i\_a$ (line 11 and 16) and the transformed loop bodies 
(lines 10--14 and lines 16--17) are executed unconditionally. 

We explain the transformation rules formally in Section~\ref{transformation}.
The transformed program can be verified by an off-the-shelf BMC. 
Note that each index will be considered in some run of the BMC since $i\_a$ is chosen non-deterministically.
Hence, if an assertion fails for any index in the original program, it fails in the transformed program too.

\section{Semantics}
\label{sec:dsa}
In this section we formalize our technique
by explaining the language and defining representation of states. 

\subsection{Language}
We formulate our analysis over a language modelled on C.
For simplicity of exposition we restrict our description to a subset of C which includes C style structures and 1-dimensional arrays.
Let \sval, \svar, and \sexpr be the sets of values computed by the program, variables appearing in the program, and expressions
appearing in the program respectively.
A value \text{$c \in \sval$} can be an integer, floating-point or boolean value.
A variable \text{$v \in \svar$}
can be a scalar variable, a structure variable, or an array variable. 
We define our program to have only one array variable denoted as \arrvar. 
However, in practice, we can handle multiple arrays in a program as explained in our technical report~\cite{Jana2016Scaling}. 
We also define \text{$\earr \subseteq \sexpr$} 
as set of array expressions of the form \text{\arrvar[\Expr]}.
A lval \Lvalue can be an array access expression or a variable. 
Let \text{$c \in \sval$}, 
\text{$\var,i \in (\svar - \{\arrvar\})$}.
We consider assignment statements, conditional statement, loop statement, and assertion statements defined by the following grammar. 
We define the grammar of our language using the following non-terminals:
Program \Prog consists of statements \Stmt which may use 
lvalues \Lvalue and expressions \Expr. We assume that programs are type correct as per C typing rules.
\begin{equation}
\renewcommand{\arraystretch}{1.2}
\begin{array}{rl}
\Prog \rightarrow & \; \,
	\Stmt
	\\
\Stmt  \rightarrow & \; \,
	\ifkw\; (\Expr) \; \Stmt \; \elsekw \; \Stmt 
	\;\; \big\lvert \;\; \ifkw\; (\Expr) \; \Stmt 
	\;\; \big\lvert \;\; \forkw \; (i = \Expr \,;\; \Expr \,;\; \Expr) \;\Stmt 
	\;\; \big\lvert 
	\\
	&
	\;\; \Stmt \; ; \; \Stmt 
	\;\; \big\lvert \;\; \Lvalue = \Expr 
	\;\; \big\lvert \;\; \assertkw(\Expr) 
	\\
\Lvalue \rightarrow & \;\,
	a[\Expr] 
	\;\; \big\lvert \;\; \var
	\\
		\label{eq:grammar.org}
\Expr \rightarrow & \; \,
	\Expr \oplus \Expr
	\;\; \big\lvert \;\; \Lvalue
	\;\; \big\lvert \;\; c
\end{array}
\end{equation}


In practice, we analyze \textsc{Ansi}-C language programs that includes functions, pointers, composite data-structures, 
all kinds of definitions, and all control structures except multi-dimensional arrays.

\subsection{Representing Program States}
We define program states in terms of memory location and the value stored in the memory location.
We distinguish between \emph{atomic} variables (such as
scalar and structure variables) whose values can be copied atomically to a memory location, from
non-atomic variables such as arrays. Since we are considering 1-dimensional arrays, the
array elements are atomic locations. 

Function \text{$\loc(a[i])$} returns the memory location corresponding to the $i^{th}$ index of array $a$.
The memory of an input program consists of all atomic locations:
\begin{align}
	\smem & = (\svar - \{\arrvar\}) \cup \big\{ \loc(\arrvar[i]) \; \big\lvert \; 0 \leq i \leq \size(\arrvar) \big\}
\label{equation:mem.org}
\end{align}
The function \size(\arrvar) returns the highest index value for array \arrvar.

A \emph{program state} is a map $\ms : \smem \to \sval$. 
\vexpr denotes the value of expression \expr in the program state \ms. 

We transform a program by creating a pair \text{$\langle \absind,\absarr \rangle$} for the array \arrvar where \absind is the witness index 
and \absarr is the witness variable.
The memory of a transformed program with additional variables is:
\begin{align}
	\smema & = (\svar - \{\arrvar\}) \cup \{\absarr\} \cup \; \{\absind\}
\label{equation:mem.trans}
\end{align}

For a transformed program, a program state is denoted by \msa and is defined over \smema. 

We explain the relation between states in original and transformed programs using an example.
Let a program \pc have an array variable $a$ and variable $k$ holding the size of the array $a$. Let
the array contain the values \text{$c_i \in \sval$}, \text{$0 \leq i < n$}, where $n \in \sval$ is the value of size of the array.
Then, a program state, $\ms$ at any program point \emph{l} can be: 
\begin{align}
\ms &= 
	\big\{  \big(k,n\big), 
         \big(\loc \left(a[0]\right),c_0\big), 
         \big(\loc \left(a[1]\right),c_1\big),
         \ldots, 
         \big(\loc \left(a[{n-1}]\right),c_{n-1} \big)
 \big\}
	\label{eq:state.p}
\end{align}

In the transformed program \pa, let \absarr and \absind be the witness variable and witness index respectively. 
Let $l'$ be the program point in \pa that corresponds to $l$ in \pc.
Then, all possible states in the transformed program at $l'$ are,
\begin{align}
	\label{eq:set.state.pa}
\msa_{0} & = \{(k,n),(i_a,0),(x_a,c_{0})\}
	\nonumber
\\
\msa_{1} & = \{(k,n),(i_a,1),(x_a,c_{1})\}
	\nonumber
\\
\ldots
	\nonumber
\\
\msa_{{n-1}} & = \{(k,n),(i_a,{n-1}),(x_a,c_{{n-1}})\}
	\nonumber
\end{align}

We now formally define how a state at a program point in the transformed program represents a state at the
corresponding program point in the original program. 

\begin{definition}
\label{def:abselement1}
Let \ms be a state
at a program point in \pc and let \msa be a state at the corresponding program point in \pa.
Then, \msa represents \ms, denoted as \text{$\msa \rightsquigarrow \ms$} if 
\[
 \msa = \left\{
		(\absind,c_1) 
  		(\absarr,c_2) 
	\right\} 
	\;\cup\;
	\left\{
		(y,c) \mid (y,c) \in \ms, 
   			y \in (\svar-\{\arrvar\})
	\right\}
	\Rightarrow \; \big(\loc(a[c_1]),c_2\big) \in \ms \; 
\]
\end{definition}


Let \assert be the assertion at line n in program P. Let \ms be a state reaching \assert in the original program with pair
\text{$\left(\loc\left(a\left[\eval{e_1}{\ms}\right]\right),\eval{e_2}{\ms}\right)$}. 
Let \msa be the state in 
transformed program, \msa represents \ms. Thus, \msa has two pairs,  
\text{$\left(\absind, \eval{e_3}{\msa}\right)$} and
\text{$\left(\absarr,\eval{e_4}{\msa}\right)$}
such that
\text{$\eval{e_3}{\msa} = \eval{e_1}{\ms}$} and \text{$\eval{e_4}{\msa} = \eval{e_2}{\ms}$} .
Hence, if the assertion \assert holds in transformed program it holds in the original program too.



\begin{figure}[!t]
\addtocounter{figure}{1}
\newcommand{\DEFRule}{%
\arrayrulecolor{lightgray}\hline\arrayrulecolor{black}
}
\begin{align*}
\T(\Expr)  & = 
	\\
	&\hspace*{-10mm}
		\Expr \equiv (\Expr_1 \oplus \Expr_2)
		\hspace*{-5mm}&\Rightarrow\hspace*{2mm}
		&
		\emit\left(\T(\Expr_1) \; \oplus \; \T(\Expr_2) \right)
		\tag{\thefigure.$E_1$}\label{transform.expr.1}
		\\
	&\hspace*{-10mm}
		\Expr \in \earr, \;
		\Expr \equiv \arrvar[\Expr_1] \;
		\hspace*{-5mm}&\Rightarrow\hspace*{2mm}
		&
		\emit\left(\left(\Expr_1 ==\absind\right)  \qkw \; \absarr \; \ckw \; nd() \right)
		\tag{\thefigure.$E_2$}\label{transform.expr.2}
		\\
	&\hspace*{-10mm}
		\text{otherwise}
		\hspace*{-5mm}&\Rightarrow\hspace*{2mm}
		& \; \emit\left(\Expr   \right)
		\tag{\thefigure.$E_3$}\label{transform.expr.3}
	\rule[-.5em]{0em}{1em}
	\\ 
	\hline
	\rule{0em}{1em}
\T(\Stmt)  & = 
	\\
	&\hspace*{-10mm}
		\Stmt \equiv (\Lvalue=\Expr), \; 
		\Lvalue \equiv \arrvar[\Expr_1] \;
		\hspace*{-5mm}&\Rightarrow\hspace*{2mm}
		&
		\emit\big(\left(\Expr_1 ==\absind\right) \qkw \; 
			\\
	& & & \phantom{\emit\big(} \absarr 
			=  \T(\Expr) \; \ckw \; \T(\Expr) \big)
		\tag{\thefigure.$S_1$}\label{transform.stmt.1}
		\\
	&\hspace*{-10mm}
		\Stmt \equiv (\Lvalue=\Expr),
		\Lvalue \not\equiv a[\Expr_1] \;
		\hspace*{-5mm}&\Rightarrow\hspace*{2mm}
		& \;\emit\left(\Lvalue=\T(\Expr)\right)
		\tag{\thefigure.$S_2$}\label{transform.stmt.2}
		\\
	&\hspace*{-10mm}
		\begin{array}[t]{@{}l@{\,}l}
		\Stmt \equiv & (\forkw(i = \Expr_1;\;\Expr_2;\;\Expr_3) \; \Stmt_1), 
		\\
		& 
		\fullarrayaccess(\Stmt), 
		\\
		& u \in \loopdefs(\Stmt_1)
		\end{array}
		\hspace*{-4.5mm}&\Rightarrow\hspace*{2mm}
		&
		\;\begin{array}[t]{@{}l@{\,}l}
		\emit\;\big( &  u=nd();
			\hspace*{8mm} // \forall u \in \loopdefs(\Stmt_1) \\
		& i=\absind;\\ 
			 & \T(\Stmt_1);\\  
			& u=nd(); 
			\hspace*{8mm} // \forall u \in \loopdefs(\Stmt_1) \\ 
		\phantom{\emit}	\; \big)
		\end{array}
		\tag{\thefigure.$S_3$}\label{transform.stmt.3}
		\\
	&\hspace*{-10mm}
		\begin{array}[t]{@{}l@{\,}l}
		\Stmt \equiv & (\forkw(i = \Expr_1;\;\Expr_2;\;\Expr_3) \; \Stmt_1),  
		\\
		& 
		\neg\fullarrayaccess(\Stmt),
		\\
		& u \in \loopdefs(\Stmt_1)
		\end{array}
		\hspace*{-5mm}&\Rightarrow\hspace*{2mm}
		&
		\;\begin{array}[t]{@{}l@{\,}l}
		\emit\;\big( & \ifkw(nd(0,1))  \\ 
		& \cskw  \; \; u=nd();
		 	\hspace*{5mm} // \forall u \in \loopdefs(\Stmt_1) \\ 
		& \;\; \;  i=nd(\loopbound(\Stmt));\\
			& \;\;\; \T(\Stmt_1); \\ 
		& \cekw \\ 
		& \; u=nd();
			\hspace*{7mm} // \forall u \in \loopdefs(\Stmt_1) \\ 
		 \phantom{\emit}\;\big)
		\end{array}
		\tag{\thefigure.$S_4$}\label{transform.stmt.4}
		\\
	&\hspace*{-10mm}
		\Stmt \equiv (\ifkw(\Expr) \; \Stmt_1 \; \elsekw \; \Stmt_2)
		\hspace*{-5mm}&\Rightarrow\hspace*{2mm}
		& \;\emit\big(\ifkw(\T(\Expr)) 
			\\
		& & & \phantom{\emit\big(} \; \T(\Stmt_1) \; \elsekw \; \T(\Stmt_2) \big)
		\tag{\thefigure.$S_5$}\label{transform.stmt.5}
		\\
	&\hspace*{-10mm}
		\Stmt \equiv (\ifkw(\Expr) \; \Stmt_1 )
		\hspace*{-5mm}&\Rightarrow\hspace*{2mm}
		& \;\emit\left(\ifkw(\T(\Expr)) \; \T(\Stmt_1)  \right)
		\tag{\thefigure.$S_6$}\label{transform.stmt.6}
		\\
	 &\hspace*{-10mm}
		\Stmt \equiv (\Stmt_1;\Stmt_2)
		\hspace*{-5mm}&\Rightarrow\hspace*{2mm}
		& \;\emit\left(\T(\Stmt_1);\T(\Stmt_2) \right)
		\tag{\thefigure.$S_7$}\label{transform.stmt.7}
		\\
	&\hspace*{-10mm}
		\Stmt \equiv (\assertkw(\Expr))
		\hspace*{-5mm}&\Rightarrow\hspace*{2mm}
		& \;\emit\left(\assertkw(\T(\Expr)) \right)
		\tag{\thefigure.$S_8$}\label{transform.stmt.8}
		\\
	&\hspace*{-10mm}
		\text{otherwise}
		\hspace*{-5mm}&\Rightarrow\hspace*{2mm}
		& \; \emit\left(\Stmt\right)
		\tag{\thefigure.$S_9$}\label{transform.stmt.9}
	\rule[-.5em]{0em}{1em}
	\\ 
	\hline
	\rule{0em}{1em}
\T(\Prog)  & = 
	\\
	&\hspace*{-10mm}
		\Prog \equiv \Stmt 
		\hspace*{-5mm}&\Rightarrow\hspace*{2mm}
		& 
		\;\begin{array}[t]{@{}l@{\,}l}
		\emit\;\big( & 
			\absind = nd\left(\size\left(a\right)\right);
				\\
			& \T(\Stmt)
			\\
		 \phantom{\emit}\;\big)
		\end{array}
		\tag{\thefigure.$P$}\label{transform.prog.1}
\end{align*}



\addtocounter{figure}{-1}
\caption{Program transformation rules. Non-terminals \Prog, \Stmt, \Expr, \Lvalue represent the code fragment in the input program
derivable from them.}
\label{fig:transform}
\end{figure}

\section{Transformation}
\label{transformation}

The transformation rules are given in Figure~\ref{fig:transform}. A transformed program
satisfies the following grammar 
derived from that of the original program (grammar \ref{eq:grammar.org}).
Let \text{$x,\absarr,\absind \in \svar$} denote scalar variable, witness variable, and witness index, respectively.
Let \text{$c,l,u \in \sval$} be values. Then,
\begin{equation}
\renewcommand{\arraystretch}{1.2}
\begin{array}{rl}
\Prog \rightarrow & \; \,
	\Init \;; \; \Stmt
	\\
\Init \rightarrow & \; \,
	\absind = nd (l,u) 
	\\
\Stmt  \rightarrow & \; \,
	\ifkw\; (\Expr) \; \Stmt \; \elsekw \; \Stmt 
	\;\; \big\lvert \;\; \ifkw\; (\Expr) \; \Stmt 
	\;\; \big\lvert \;\; \Stmt \; ; \; \Stmt 
	\;\; \big\lvert \;\; 
	\Lvalue = \Expr 
	\;\; \big\lvert \;\; \assertkw(\Expr) 
	\\
\Lvalue \rightarrow & \;\,
	\var \;\; \big\lvert \;\; \absarr \;\; \big\lvert \;\; \absind
	\\
\Expr \rightarrow & \; \,
	\Expr \oplus \Expr
	\;\; \big\lvert \;\; \Lvalue
	\;\; \big\lvert \;\; c
	\;\; \big\lvert \;\; nd()
	\;\; \big\lvert \;\; nd(l,u)
\end{array}
\label{eq:grammar.trans}
\end{equation}
The non-terminal \Init represents the initialization statements for witness index.
Witness variable is initialized in the scope same as that in the original program.

We use the functions described below in the transformation rules.
\begin{compactitem}
\item Function $nd$ returns a non-deterministically chosen value from the given range $l,u$; $l$ and $u$ being the lower and upper limit respectively. When range is not provided, $nd$ 
	returns a non-deterministic value based on the type of \Lvalue. 
\item Function \transform takes the text derived from a non-terminal and transforms it. Function \emit shows the actual
      code that would be emitted. We ignore the details of number of parameters and the type of the parameters of \emit.
      We assume that it takes the code emitted by \transform and possibly some additional statements and outputs the combined code.
      It has been used only to distinguish the transformation time activity and run time activity.
      For example, the boolean conditions in cases~\ref{transform.expr.2} and~\ref{transform.stmt.1} are not evaluated by the
      body of function \transform but is a part of the transformed code and is evaluated at run time when the transformed program is executed.
	Similar remarks apply to the \ifkw statements and other operations inside the parenthesis of \emit function.
\item Function $\fullarrayaccess(\Stmt)$ analyzes\footnote{Analysis can be over-approximated.} the characteristics of the loop \Stmt.
	\begin{itemize}
		\item When the loop \Stmt accesses array \arrvar completely, $\fullarrayaccess(\Stmt)$ returns true.
		This means that loop either reads or write all the indices of the array.
	\item When the loop \Stmt accesses array \arrvar partially, $\fullarrayaccess(\Stmt)$ returns false. This means
		that the loop may not access all the indices or some indices are being read while some other indices are being written.
	\item When loop \Stmt do not access an array, $\fullarrayaccess(\Stmt)$ returns false. 
	\end{itemize}


\item Function $\loopdefs(\Stmt)$ returns the over-approximated set of variables modified in the loop \Stmt.
	\begin{compactitem}	
	\item Scalar variables are included in this set if they appear on the left hand side of any assignment statement in \Stmt 
	(except when the RHS is a constant). 
	\item Loop iterator variable $i$ of loop \Stmt is not included in this set.
	\item 
	Array variable \arrvar 
	is included in this set when the array access expression appears on the left hand side of an assignment
	and the value of index expression is different from the current value of the loop iterator $i$.
\end{compactitem}

\item Function $\size(a)$ returns the highest index value for array \arrvar. 
\end{compactitem}
With the above functions, the transformation rules are easy to understand. Here we explain non-trivial transformations.
\begin{compactitem}
\item To choose an array index for a run, witness index (\absind) is initialized at the start of the program 
	to a non-deterministically chosen value from the range of the indices of the array (case~\ref{transform.prog.1}). 
	This value determines the array element ($a[\absind]$) represented by the witness variable (\absarr). 
\item An array access expression in LHS or RHS is replaced by the witness variable (\absarr) provided the
      values of the witness index and index expression of the array access expression match.
      If the values do not match, it implies that the element accessed 
      is not at the non-deterministically chosen index \absind.
      Hence for any other index the assignment does not happen (case~\ref{transform.stmt.1}). 
      Similarly, when any other index is read in RHS, it is
      replaced with a non-deterministic value (case~\ref{transform.expr.2}). 
\item Loop iterations are eliminated by removing the loop header containing initialization, test, and increment expression for 
	loop iterator variable. The loop bodies are transformed as follows :
	\begin{compactitem}
	\item Each variable in the set returned by $\loopdefs(\Stmt)$ is assigned a non-deterministic value 
		at the start of the loop body and also after the loop body.
	These assignments ensure that values dependent on loop iterations are 
	over-approximated when used inside or outside the loop body. 
       \item The loop iterator $i$ is a special scalar variable. 
	 A loop \Stmt where \fullarrayaccess(\Stmt) holds (case~\ref{transform.stmt.3}) essentially means that loop bound is same as the array size 
	 and array is accessed using loop iterator as index.
	 Hence it is safe to replace array access with \absarr where the 
	 values of loop iterator and index expression match. 
	 To ensure this we equate loop iterator with \absind.
	 This models the behaviour of the original program precisely.
	 However, when $\fullarrayaccess(\Stmt)$ does not hold (case~\ref{transform.stmt.4}), we assign loop iterator $i$ to a 
	 non-deterministically chosen value from the loop bound. 
 \item Each statement in the loop body is transformed 
	as per the transformation rules. 
 \item Finally, the entire loop body is made conditional using a non-deterministically chosen true/false value 
	 when $\fullarrayaccess(\Stmt)$ does not hold. 
	This models the partial
        accesses of array indices which imply that some of the values defined before the loop may reach after the loop.
	However, the transformed loop body is unconditionally executed when $\fullarrayaccess(\Stmt)$ holds.
      \end{compactitem}
 \end{compactitem}

\section{Soundness}
\label{sec:proof}

\externaldocument{transformation}

This section outlines the claim that the proposed transformation is sound, i.e. if the transformed
program is safe, then so is the original program. As discussed in Section~\ref{sec:dsa}, the
soundness is immediate if the abstract states ``represent'' the
original states. We, therefore, prove that the proposed transformations ensure
that the \emph{represents} relation, $\rightsquigarrow$, holds between abstract and
original states. 
For the base case, we prove that $\rightsquigarrow$ holds in the beginning -
before applying any transformation (Lemma~\ref{lemm:program.start}). 
In the inductive step, we prove
that if $\rightsquigarrow$ holds at some stage during the transformation, then the subsequent
transformation continues to preserve $\rightsquigarrow$ (Lemma~\ref{lemma:strucinduc.abs}). We prove this by structural
induction on program transformations. We prove that each transformed expression is over-approximated when $\rightsquigarrow$ holds in (Lemma~\ref{lemm:rhs1}).
Detailed proof is provided in our technical report~\cite{Jana2016Scaling}.


\newcommand{\proofOutline}{\noindent\text{\em Proof Outline.\ }}

\begin{lemma}
\label{lemm:program.start}
Let the start of the original program (i.e. the program point just before the code derivable from non-terminal \Stmt in
production \text{$\Prog \rightarrow \Stmt$} in grammar defined in equation(~\ref{eq:grammar.org}) be denoted by $l$. 
The corresponding program point in the 
transformed program \pa, denoted by $l'$, is just after \Init and just before the non-terminal $S$ in production \text{$\Prog \rightarrow \Init \;; \; \Stmt$}
(Grammar in equation~\ref{eq:grammar.trans}).
Let \ms and \msa be the states at $l$ and $l'$ in \pc and \pa respectively. 
Then,
\text{$\msalp \rightsquigarrow\msl$}.
\end{lemma}
\proofOutline
Since the initial values of non array variables are preserved, the initial value of the element of array $a[\absind]$ is assigned to \absarr, and
\absind is non-deterministically chosen, the lemma holds.

\begin{lemma}
\label{lemm:rhs1}
Let $\msl$ be a state at a program point $l$ in \pc 
and $\msalp$ be a state at the corresponding program point $l'$ in transformed program \pa.
Consider an arbitrary expression \text{$e \in \sexpr$} just after $l$ in original program \pc.
Then, 
\[
\text{$\msalp \rightsquigarrow\msl \Rightarrow \eval{\T(e)}{\msalp} \supseteq \eval{e}{\msl}$}.
\]
\end{lemma}

\proofOutline
Since $e$ is derived from \Expr (grammar~\ref{eq:grammar.org}),
the over-approximation of values can be proved by structural induction on the productions for \Expr.

\begin{lemma} 
Let $l$ and $m$ be the program points just before and after a statement $s$ in \pc and let $\msl$ and $\msm$ be 
the states at $l$ and $m$ respectively.
Let $l'$ and $m'$ be the program points just before and after the corresponding transformed statement $\T(s)$ in \pa.
Let \msalp and \msamp be the states at $l'$ and $m'$ respectively.
Then, \text{$\msalp \rightsquigarrow \msl \Rightarrow \msamp \rightsquigarrow \msm$}.


\label{lemma:strucinduc.abs}
\end{lemma}

\proofOutline
Since statement $s$ is derived from non-terminal \Stmt in the grammar~\ref{eq:grammar.org}
the lemma can be proved by structural induction on \Stmt.

\begin{theorem}
If the assertion \assert is violated in the original program \pc, then it will be violated in transformed program \pa also. 
\end{theorem}
\begin{proof}
Let the assert get violated for some $a[c]$. Since $\absind$ is initialized
non-deterministically it can take the value $c$ and we have shown in Lemma~\ref{lemm:rhs1} that
all expressions in $\pa$ are over-approximated. Lemma~\ref{lemm:program.start} and Lemma~\ref{lemma:strucinduc.abs} ensure the premise for Lemma~\ref{lemm:rhs1}.
Hence the theorem follows.

\end{proof}

\newcommand{\Defset}{\text{\sf\em$\mathbb{S}_\text{\em def}$}\xspace}
\newcommand{\ImpVarset}{\text{\sf\em$\mathbb{V}_\text{\em imp}$}\xspace}
\newcommand{\ImpExpset}{\text{\sf\em$\mathbb{E}_\text{\em imp}$}\xspace}

\section{Precision}
\label{sec::prec}

We characterize the assertions for which our transformation is precise -- an assertion will fail in 
\pa if and only if it does so in \pc.
We denote such an assertion as \assertinv. We focus on \assertinv in a loop.
A program can have array accesses outside loops too. In such cases we do not claim precision; as per our experience such situations
are rare in programs with large-sized arrays.

Our transformations
replace array access expressions and loop statements while the statements involving
scalars alone outside the loop remain unmodified. Hence precision criteria need to focus on the statements within loops and not outside it.

Let assertion \assertinv be in loop statement $S_\assertinv$.
Let \ImpVarset be the set of variables and \ImpExpset be the set of array access expressions on which \assertinv is 
data or control dependent within the loop $S_\assertinv$.
Let the set of loop statements from where definitions reach \assertinv be denoted by \Defset, note that this
set is a transitive closure for data dependence.
Our technique is precise when:
\begin{itemize}
\item $\fullarrayaccess(S)$ holds for each \text{$S \in \{ S_\assertinv\} \cup \Defset$} \hfill(rule $l_1$)
\item If  $a[e] \in \ImpExpset$ then
	\begin{itemize}
	\item the index expression $e = i$ where $i$ is the loop iterator of loop $S_\assertinv$  \hfill 
		(rule $a_2$)
 	\item $\arrvar \notin \loopdefs(S)$ where \text{$S \in \{ S_\assertinv\} \cup \Defset$} \hfill (rule $a_3$)
\end{itemize}
\item  If $\var \in \ImpVarset$ then
	$\var \notin \loopdefs(S)$ where \text{$S \in \{ S_\assertinv\} \cup \Defset$} \hfill (rule $s_4$)
\item For an assignment statement of the form $a[e_1] = e_2$ in loop $S$ where $S \in \Defset$,  
      \begin{itemize}
       \item if $e_2$ is an array access expression then it must be of the form $a[i]$ where $i$ is the loop iterator 
	       of loop $S$ 	\hfill (rule $d_5$)
	\item if $e_2$ is \var then $\var \notin \loopdefs(S)$ where \text{$S \in \Defset$} \hfill 
		(rule $d_6$)
	\end{itemize}

\end{itemize}

\begin{theorem}
\label{th:comp}
If the assertion \assertinv; that satisfies above rules; holds in the original program \pc, then it will hold in the 
transformed program \pa also.

\end{theorem}

\begin{proof}
	The transformed program is over-approximative because our transformation rules (~\ref{transform.stmt.3},~\ref{transform.stmt.4},~\ref{transform.expr.2}) 
	introduce non-deterministic values.
	We prove this theorem by showing that if assertion is of the form \assertinv then 
	none of these transformation rules introduce non-deterministic values in the transformed program.

\begin{itemize}
	\item Since rule $l_1$ holds unconditionally, case~\ref{transform.stmt.4} will not apply. 
	Hence no extra paths are added in transformed program.
		Also, since case~\ref{transform.stmt.3} applies, assignment $i=\absind$ will be added for 
			$S_\assertinv$ and the loop statements in \Defset.
	\item When rule $a_2$ holds, since rule $l_1$ holds $a[e]$ get replaced by \absarr always (case~\ref{transform.expr.2}).
	\item When rule $a_3$ holds, assignment $\absarr=nd()$ is not added (case~\ref{transform.stmt.3}).
	\item When rule $s_4$ holds, assignment $\var=nd()$ is not added (case~\ref{transform.stmt.3}).
	\item When rule $d_5$ holds, since rule $l_1$ holds $a[e]$ in RHS gets replaced with \absarr (case~\ref{transform.expr.2}).
	\item When rule $d_6$ holds, scalars in RHS are not assigned with a non-deterministic value.
	
\end{itemize}
\end{proof}

Note that rule $s_4$ is a very strong condition to ensure that non-deterministic values do not reach \assertinv.
We can relax this rule 
when $\var \in \loopdefs(S_\assertinv)$ under these two conditions: 
\begin{itemize}
\item definition of \var appears before the assert statement in the loop
\item \var is defined with a constant or using loop iterator $i$ only.
\end{itemize}

None of the transformation rules replace variable \var. Definition of \var to non-deterministic value ($\var = nd()$) 
gets re-defined by original assignment (retained in the transformed loop body) appearing before the assert statement.
Since \var is defined with a constant or $i$ ($i=\absind$ is added for $S_\assertinv$),
its value is not over-approximated.

\begin{table}[t]
\centering
\caption{Results on SV-COMP Benchmark Programs.}
\label{svcomp_table}
\begin{tabular}{|c|c|c|c|c|c|}
\hline
\#programs = 118 & \begin{tabular}[c]{@{}c@{}}\#correct\\ true\end{tabular} & \begin{tabular}[c]{@{}c@{}}\#correct \\ false\end{tabular} & \begin{tabular}[c]{@{}c@{}}\#incorrect \\ true\end{tabular} & \begin{tabular}[c]{@{}c@{}}\#incorrect \\ false\end{tabular} & \#no result \\ \hline \hline
Expected Results & 84 & 34 & - & - & 0 \\ \hline \hline
$CBMC_\alpha$ & 47 & 6 & 6 & 0 & 59 \\ \hline
$CBMC_\beta$ & 9 & 5 & 0 & 0 & 104 \\ \hline
Transformation+$CBMC_\beta$ & 25 & 34 & 0 & 59 & 0 \\ \hline \hline
\multicolumn{6}{|c|}{$CBMC_\alpha$ - SV-COMP2016 (unsound) CBMC, $CBMC_\beta$ - sound CBMC 5.4} \\ \hline
\end{tabular}
\end{table}

\section{Experimental Evaluation}
\label{exp}
We have implemented our transformation engine using static
analysis\footnote{PRISM, a static analyzer generator developed at TRDDC, Pune
~\cite{chimdyalwar2011effective,khare2011static}}.  It supports
\textsc{Ansi}-C programs with 1-dimensional arrays. The experiments are
performed on a 64-bit Linux machine with 16 Intel Xeon processors running at
2.4GHz, and 20GB of RAM. More details of optimization and implementation, including handling of multiple arrays,
are provided in our
technical report~\cite{Jana2016Scaling}.

Our transformation engine outputs C programs. Although we could take
any off-the-shelf BMC for C program to verify the transformed code, we use CBMC
in our experiments as it is known to handle all the constructs of
\textsc{Ansi}-C.  We discuss the results of our experiments on academic
benchmarks and industry codes.  For want of space, we omit the results of
various BMCs on patterns from industry code; those results are shared in our
technical report~\cite{Jana2016Scaling}.


\subsection{Experiment 1 : SV-COMP Benchmarks} SV-COMP
benchmarks~\cite{svcompbench} contain an established set of programs under
various categories intended for comparing software verifiers.
Results for \emph{ArraysReach}\footnote{Programs in \emph{ArrayMemSafety}
access arrays without using index and cannot be transformed.} from the
\emph{array} category
for CBMC used in SV-COMP 2016 ($CBMC_\alpha$), CBMC 5.4 ($CBMC_\beta$) and CBMC
5.4 on transformed programs (Transformation+$CBMC_\beta$) are
consolidated\footnote{Case by case results available at \\
\htmladdnormallink{https://sites.google.com/site/datastructureabstraction/home/sv-comp-benchmark-evaluation-1}{https://sites.google.com/site/datastructureabstraction/home/sv-comp-benchmark-evaluation-1}}
in Table~\ref{svcomp_table}.
\emph{ArraysReach} has 118 programs. $CBMC_\alpha$, an unsound version of CBMC,
gave correct results for 53 programs. 
However, $CBMC_\beta$ gave correct results for 14 programs.
We compare the results of Transformation+$CBMC_\beta$ on three criteria:

\begin{compactitem}
\item Scalability: it scaled up for all 118 programs.
\item Soundness: it gave sound results for all 118 programs. For the 6
program for which $CBMC_\alpha$ gave unsound results, our results are not only
sound but are also precise.
\item Precision: it gave precise results for 59 programs. Out of these
$CBMC_\alpha$ ran out of memory for 45 programs ($CBMC_\alpha$ ran out of
memory for 14 additional programs).  On the other hand, 22 true
programs reported correctly by $CBMC_\alpha$ were verified as false by
Transformation+$CBMC_\beta$.  Transformation+$CBMC_\beta$ verified 25 program
as true which did not include 8 of programs reported correctly as true by
$CBMC_\beta$.	
\end{compactitem}

Our technique is imprecise for the other 59 of 118 programs as they do not comply with
the characterization of precision provided in Section~\ref{sec::prec}.  As can
be seen, there is a trade-off between scalability and precision.  From the view
point of reliability of results, soundness is the most desirable property of a
verifier. Our technique satisfies this requirement. Further, it not only scales up  but is also precise implying its
practical usefulness.

\subsection{Experiment 2 : Real-life Applications}
\begin{table}[t]
\centering
\caption{Real-life Application Evaluation}
\label{RealApp_table}
\begin{tabular}{ccccccccccc}
\hline
\multicolumn{4}{|c|}{Application details} & \multicolumn{3}{c|}{Sliced+CBMC} & \multicolumn{3}{c|}{\begin{tabular}[c]{@{}c@{}}Sliced\\ +Transformation \\ +CBMC\end{tabular}} & \multicolumn{1}{c|}{\multirow{2}{*}{\begin{tabular}[c]{@{}c@{}}\%\\ False\\  Positive \\ Reduction\end{tabular}}} \\ \cline{1-10}
\multicolumn{1}{|c|}{Name} & \multicolumn{1}{c|}{\begin{tabular}[c]{@{}c@{}}Size\\ (LoC)\end{tabular}} & \multicolumn{1}{c|}{ $\% loop^{full}$} & \multicolumn{1}{c|}{\#Asserts} & \multicolumn{1}{c|}{\#P} & \multicolumn{1}{c|}{\#F} & \multicolumn{1}{c|}{\#T} & \multicolumn{1}{c|}{\#P} & \multicolumn{1}{c|}{\#F} & \multicolumn{1}{c|}{\#T} & \multicolumn{1}{c|}{} \\ \hline \hline
\multicolumn{1}{|c|}{\it{navi1}} & \multicolumn{1}{c|}{1.54M} & \multicolumn{1}{c|}{100} & \multicolumn{1}{c|}{63} & \multicolumn{1}{c|}{0} & \multicolumn{1}{c|}{0} & \multicolumn{1}{c|}{63} & \multicolumn{1}{c|}{52} & \multicolumn{1}{c|}{1} & \multicolumn{1}{c|}{10} & \multicolumn{1}{c|}{82.5} \\ \hline
\multicolumn{1}{|c|}{\it{navi2}} & \multicolumn{1}{c|}{3.3M} & \multicolumn{1}{c|}{93.4} & \multicolumn{1}{c|}{103} & \multicolumn{1}{c|}{0} & \multicolumn{1}{c|}{0} & \multicolumn{1}{c|}{103} & \multicolumn{1}{c|}{95} & \multicolumn{1}{c|}{1} & \multicolumn{1}{c|}{7} & \multicolumn{1}{c|}{92.2} \\ \hline
\multicolumn{1}{|c|}{icecast\_2.3.1} & \multicolumn{1}{c|}{336K} & \multicolumn{1}{c|}{59.1} & \multicolumn{1}{c|}{114} & \multicolumn{1}{c|}{0} & \multicolumn{1}{c|}{0} & \multicolumn{1}{c|}{114} & \multicolumn{1}{c|}{53} & \multicolumn{1}{c|}{61} & \multicolumn{1}{c|}{0} & \multicolumn{1}{c|}{46.5} \\ \hline \hline
\multicolumn{11}{|c|}{\begin{tabular}[c]{@{}c@{}} $loop^{full}$ - loop \Stmt where \fullarrayaccess(\Stmt) holds, \\ P - Assertion Proved, F - Assertion Failed, T - Timeout\end{tabular}} \\ \hline 

\end{tabular}
\end{table}

We applied our technique on 3 real-life applications - 
{\it navi1} and {\it navi2} are industry codes implementing the navigation
system of an automobile and icecast\_2.3.1 is an open source project for
streaming media~\cite{icecast}.  We appended assertions using \emph{null
pointer dereference} (NPD) warnings from a sound static analysis\footnote{TCS
Embedded Code Analyzer (TCS ECA)\\
\htmladdnormallink{http://www.tcs.com/offerings/engineering\_services/Pages/TCS-Embedded-Code-Analyzer.aspx}{http://www.tcs.com/offerings/engineering\_services/Pages/TCS-Embedded-Code-Analyzer.aspx}}
tool as follows.  Lets say the dereference expression is $*a[i].p$.  A
statement $assert(a[i].p!=null)$ is added in the code just before statement
containing dereference expression.

We ran CBMC on these applications with a time out of 30 minutes. CBMC did not
scale on the original as well as the sliced programs. We ran our transformation engine on
sliced programs. Table~\ref{RealApp_table} shows the consolidated results of
our experiments. Out of 280 assertions, sliced+transformation+CBMC proved 200
assertions taking 12 minutes on average for transformation+verification. This
is a much less in comparison to the time given to CBMC for sliced programs
(sliced+CBMC), which was 30 minutes.  

To verify the correctness of our implementation, we analyzed the warnings
manually. We found that all 280 warnings were false, implying that all the
assertions should have been proved successfully.
\begin{compactitem}
\item CBMC could scale up for such large applications because there are no loops in transformed programs.
However, CBMC could not scale for 17 cases even after transformation because of the presence of a long recursive call chain of calls through function
pointers.
\item CBMC could not prove 63 of
the assertions since array definitions reaching at the assertion were from the
loops where $\fullarrayaccess(\Stmt)$ did not hold. 
Hence the witness variable takes over-approximated values.
\item CBMC proved 200 assertions, where
all the conditions for precision mentioned in Section~\ref{sec::prec} get fulfilled. In these experiments, we checked for the NPD property which is
value-independent. Moreover, we found that the assertions inserted by us are
not control-dependent on any scalar.
\end{compactitem}

Note that the number of false warnings
eliminated in an application is proportional to the number of loops for which
\fullarrayaccess(\Stmt) hold.  
Over a diverse set of applications, we found that our technique
could eliminate 40-90\% of false warnings. This is a significant value addition
to static analysis tools that try to find defects and end up generating a large
number of warnings. In fact, our own effort grew out of the need of handling
warnings that were generated by our proprietary static analysis tool, a large
fraction of which were false positives.

\section{Related Work}
\label{relwork}

The literature on automated reasoning about array-manipulating code can be broadly categorized into \emph{analysis} and \emph{verification}.
Most methods that analyze programs manipulating arrays~\cite{blanchet2002design,gopan2005framework,halbwachs2008discovering,cousot2011parametric,liu2014abstraction}
are based on abstract interpretation.
Cornish et al.~\cite{cornish2014analyzing} transform a program to remove arrays and discover 
non-trivial universally quantified loop invariants by analyzing
the transformed program using off-the-shelf abstract scalar analysis. 
Since they create additional blocks for each value
of \emph{summary variable}, the program size increases considerably raising concerns about scalability. 
Similar to our approach, Monniaux et al.~\cite{monniaux2015simple} transform array programs by replacing array operations with a scalar. However they keep loops. These programs are then analyzed using methods producing invariants (\emph{back-ends}).
CBMC did not scale up 
on the transformed ''array copy''
example (10000 loop bound) given in the paper, suggesting that scalability is a concern with this technique too.
However, using our technique CBMC scaled for the same program.
 
Dillig et al.~\cite{dillig2010fluid} introduced fluid updates of arrays in order
to do away with strong and weak updates. Their technique uses indexed locations
along with bracketing constraints, a pair of over- and under-approximative
constraints, to specify the concrete elements being updated. 
In another
work~\cite{dillig2011precise}, they propose an automatic technique
to reason about contents of arrays (or containers, in general). However, they
introduce an abstraction to encode all values that (a subset of) elements may
have. In contrast, since our technique choses only one representative element
to work with, we can capture its value precisely.

Template-base methods~\cite{beyer2007invariant,Gulwani2008Lifting} have been
very useful in synthesizing invariants but these techniques are ultimately
limited by a large space of possible templates that must be searched to get a
good candidate template. This has also led to semi-automatic approaches, such
as~\cite{flanagan2002predicate}, where the predicates are usually suggested by
the user. Our approach, however, is fully automatic and proves safety by
solving a bounded model checking instance instead of computing an invariant
explicitly.

Verification tools based on CEGAR have been applied successfully to
certain classes of programs, e.g., device drivers~\cite{ball2002s}. However,
this technique is orthogonal to ours. In fact, a refinement framework in
addition to our abstraction would make our technique complete. Several other
techniques have been used to scale BMCs to tackle complex, real-world
programs such as acceleration~\cite{kroening2013under} and
loop-abstraction~\cite{darke2015over}.  But these techniques are not shown to
be beneficial in abstracting complex data structures.
Booster~\cite{alberti2014booster}, a recent tool for verifying C-like programs
handling arrays, integrates acceleration and lazy abstraction with interpolants
for arrays~\cite{alberti2014decision,alberti2012lazy}.  It exploits
acceleration techniques to compute an exact set of reachable states, whenever
possible, for programs with arrays.  For instance, their technique works on
\emph{$simple_{\mathcal{A}}^0$} programs~\cite{alberti2014decision}. However,
there are syntactic restrictions that limit the applicability of acceleration
in general for programs handling arrays. Note that Booster uses acceleration,
instead of abstraction-based procedures, for want of a precise solution (not
involving over-approximations). Since our technique is also precise for a
characterizable class of programs, it is certainly possible to gainfully
combine the two techniques in order to handle a larger class of programs than
what either of them can handle in isolation.

\section{Conclusions and Future Work}
\label{conclusion}

Verification of programs with loops iterating over arrays is a challenging problem because of large sizes of arrays.
We have explored a middle ground between the two
extremes of relying completely on dynamic approaches of using model checkers on the one hand and using completely
static analysis involving complex domains and fix point computations on the other hand. Our experience shows that using static analysis
to transform the program and letting the model checkers do the rest is a sweet spot that enables verification of
properties of arrays using an automatic technique that is generic, sound, scalable, and reasonably precise.

Our experiments show that the effectiveness of our technique depends on the characteristics of programs and
properties sought to be verified.
We are able to eliminate 40-90\% of false warnings from
diverse applications.
This is a significant value addition to static analysis that try to find
defects and end up generating a large number of warnings which need to be resolved manually for safety critical applications. 
Our effort grew out of our own experience of such manual reviews which showed
a large number of warnings to be false positives.

We plan to make our technique more precise by augmenting it with a refinement step to verify the programs that are 
reported as unsafe by our current technique. 
Finally, we wish to extend our technique on other data structures such as maps or lists.

\Urlmuskip= 0mu plus 2mu\relax
\bibliographystyle{abbrv}
\bibliography{refs}

\end{document}